\documentstyle[amsfonts,12pt]{article}

\title{}\date{}

\title{Existence of Optical Vortices}

\author{Yisong Yang\\Department of Mathematics\\Polytechnic Institute of New York University\\Brooklyn, New York 11201, USA\\\\Ruifeng Zhang\\Institute of Contemporary Mathematics\\School of Mathematics\\Henan University\\
Kaifeng, Henan 475004, PR China}
\newcommand{\bfR}{{\Bbb R}}

\newcommand{\bfZ}{{\Bbb Z}}

\def\XXint#1#2#3{{\setbox0=\hbox{$#1{#2#3}{\int}$}
 \vcenter{\hbox{$#2#3$}}\kern-.5\wd0}}

\newtheorem{oldtheorem}{Theorem}[section]
\newtheorem{oldassertion}[oldtheorem]{Assertion}
\newtheorem{oldproposition}[oldtheorem]{Proposition}
\newtheorem{oldremark}[oldtheorem]{Remark}
\newtheorem{oldlemma}[oldtheorem]{Lemma}
\newtheorem{olddefinition}[oldtheorem]{Definition}
\newtheorem{oldclaim}[oldtheorem]{Claim}
\newtheorem{oldcorollary}[oldtheorem]{Corollary}

\newenvironment{theorem}{\begin{oldtheorem}$\!\!\!${\bf.}}{\end{oldtheorem}}

\newenvironment{lemma}{\begin{oldlemma}$\!\!\!${\bf.}}{\end{oldlemma}}

\newbox\qedbox
\newenvironment{proof}{\smallskip\noindent{\bf Proof.}\hskip \labelsep}%
                        {\hfill\penalty10000\copy\qedbox\par\medskip}

\newcommand{\dd}{\mbox{d}}
\newcommand{\ee}{\end{equation}}
\newcommand{\be}{\begin{equation}}\newcommand{\bea}{\begin{eqnarray}}
\newcommand{\eea}{\end{eqnarray}}
\newcommand{\ii}{\mbox{i}}\newcommand{\e}{\mbox{e}}
\newcommand{\pa}{\partial}
\newcommand{\vep}{\varepsilon}

\newcommand{\nn}{\nonumber}

\newcommand{\lm}{\lambda}

\begin{document}
\maketitle
\begin{abstract}
Optical vortices arise as phase singularities of the light fields and are of central interest in modern optical
physics. In this paper, some existence theorems are established for stationary vortex wave solutions of a
general class of nonlinear Schr\"{o}dinger equations. There are two types of results. The first type concerns
the existence of positive-radial-profile solutions which are obtained through a constrained minimization approach.
The second type addresses the existence of saddle-point solutions through a mountain-pass-theorem or min-max method
so that the wave propagation constant may be arbitrarily prescribed in an open interval.
Furthermore some explicit estimates for the lower bound and sign
of the wave propagation constant
with respect to the light beam power and vortex winding number are also derived for the first type solutions.
\end{abstract}

\maketitle

\section{Introduction}
\setcounter{equation}{0}

Vortices have important applications in many areas
of modern physics including condensed matter systems, particle interactions, cosmology, and superfluids. 
Research on vortices in optics also has a long history and was initiated in as early as 1964 by Chiao, Garmire, and  Townes \cite{Chiao}
who explored some conditions under which
a light beam can produce its own waveguide and propagate without spreading.
They described such phenomenon as self-trapping, attributed it to light propagation in materials whose dielectric
coefficient increases with field intensity
in the context of high-intensity light beams such as
lasers, predicted 
marked optical and physical effects,
 and suggested the occurrence of optical vortices. Such vortices have since then
been
 observed in numerous studies \cite{AB,BKK,RSS,SO,SL} and become
 a much pursued subject in optics \cite{CG,DK,Du,LD,NNK,SGV,VB} both theoretically and experimentally. (For a 
vast literature
up to 2005 and for a description of the rich features and profound applications of optical vortices, see the nice survey article by Desyatnikov, Kivshar, and  Torner \cite{DKT}. See also \cite{KMT} for a more recent survey of the subject
in a broader perspective.) As waves, light propagation may be described by a wave function.
At certain spots of space, the intensity of the waves vanishes and the phase of the waves
cannot be defined. Thus, such spots are phase singularities which were recognized in the comprehensive work of
Nye and Berry \cite{NB} as crucial characteristics of general wave motions in which vortices are present. These phase
singularities, also referred to as dislocations or defects of waves, are the centers of vortices, around which 
energy
concentrates. In the context of light waves, vortices are centered around vortexlines and light waves are twisted around
the vortexlines. The twisting arises from the phase ambiguity around a vortexline and is of a topological nature.
The twisting centers are exactly the vortex cores where light waves cancel out leading to darkness so that light intensity
measured at any cross section vertical to a light beam axis should display concentric ring-like patterns
around the dark core. Light beams of such structure are also vividly termed ``helical beams" \cite{BSV}.
In optics research,
a fundamental prototype situation is when the light waves are described by a complex-valued wave function 
governed by nonlinear Schr\"{o}dinger 
equations \cite{A,DY,KVT,K,MSZ,NAO,RLS,SK}. These theoretical studies provide a broad range of 
interesting analytic
problems related to the existence and properties of optical vortices for mathematical investigation. 

Our aim in the present
work is to obtain some existence theorems for the optical vortex solutions explored by Salgueiro and Kivshar in
\cite{SK} through a study of the normalized
nonlinear Schr\"{o}dinger equation 
\be \label{1.1}
\ii\frac{\pa\psi}{\pa z}+\nabla^2_\perp \psi+(V+s|\psi|^2)\psi=0,
\ee
where  $\psi$ is a complex-valued optical field propagating in the (longitudinal) $z$-direction, 
$\nabla^2_\perp$ is the Laplace operator over the (transverse) plane of coordinates $(x,y)$ which is
perpendicular to the $z$-axis, $V$ is an external potential function (cf. \cite{K,SMF}),
and $s=\pm1$ is the sign symbol indicating either a focusing or defocusing situation \cite{KVT} which is taken to be $+1$ (focusing) in \cite{SK} and will be our main focus
(the defocusing case $s=-1$ will be seen to be straightforward). The interest of (\ref{1.1})
actually goes beyond nonlinear optics. For example, it also arises in the study of the Bose--Einstein condensates
\cite{A,Du,KK,SMF} and is referred to as the Gross--Pitaevskii equation.
An important simplified situation that allows optical vortices to present is when $V$
depends on the radial variable only, $V=V(r), r=\sqrt{x^2+y^2}$. In this situation one may expect to find
an $n$-vortex
solution of (\ref{1.1}) of the form \cite{SK}
\be \label{1.2}
\psi=\psi(r,\theta,z)=u(r)\e^{\ii(n\theta+\beta z)},
\ee
where $r,\theta$ are polar coordinates over $\bfR^2$, $u(r)$
is the radial profile function which gives rise to the intensity of light waves, integer $n\in\bfZ$ is 
the winding number or vortex charge of the vortex solution, and $\beta\in\bfR$ is the wave propagation constant
\cite{SK}. This ansatz describes a vortex wave centered around the $z$-axis where $r=0$ and propagating along the
$z$-axis.
Inserting (\ref{1.2}) into (\ref{1.1}), we arrive at the following $n$-vortex equation
\be\label{1.3}
(ru_r)_r-\left(\frac{n^2}r+\beta r\right)u+r(V+su^2)u=0,
\ee
of cubic nonlinearity. The presence of the vortex core at $r=0$ requires $u(0)=0$. As
in \cite{SK} (for $n=1$), we are interested in ring-shaped
vortices so that light intensity concentrates around the vortex core which suggests that $u(r)$ may be assumed to
vanish at a sufficiently large distance. Mathematically, this indicates that we may impose another `boundary' condition,
say $u(R)=0$, at a certain distance $R>0$ away from the core of the
vortex as seen in the numerical results of the work \cite{SK}.
Thus, the problem of the existence of optical vortices is reduced into a two-point boundary value problem
with undetermined parameter $\beta$ and prescribed $R$, for any given $n\in\bfZ$. To tackle this problem, we shall use
the methods of calculus of variations. Our methods allow us to obtain two types of results for the nontrivial focusing
case
$s=+1$. The first type of results
rely on a constrained minimization approach. The nature of minimization leads us to obtaining positive-valued solutions in
the open interval $(0,R)$ and that the propagation constant $\beta$ arises
as a Lagrange multiplier due to the constraint. The second type of results are obtained from searching for saddle points of the action 
functional associated to the problem. We will see that, in this latter case, there is no assurance that the solutions must 
stay positive-valued but the propagation constant $\beta$ arises as a prescribed quantity. 

In the next two sections, we shall concentrate on the focusing case when $s=+1$. In Section 2, we formulate the problem of existence of optical vortices as a constrained minimization problem,
state the main existence results regarding positive solutions, and then present the proofs.
We will see that the propagation constant $\beta$ arises as a Lagrange multiplier which is ensured to be negative
when  the vortex charge $n$ is sufficiently large. We will also derive some lower estimate
for $\beta$. In Section 3, we treat $\beta$
as a prescribed quantity and use a mountain-pass theorem approach to establish the existence of solutions for any $R$
and vortex charge $n$. In particular, we show that the propagation constant $\beta$ may assume any prescribed value
in an explicitly given interval.
In Section 4, we briefly discuss the defocusing case when $s=-1$.

\section{Vortices via constrained minimization}
\setcounter{equation}{0}

As described in the previous section, we will be interested in `ring vortices' such that (\ref{1.3}) is considered over a bounded interval $(0,R)$
($R>0$) so that $u$ vanishes at the two endpoints of the interval. As mentioned earlier,
we will mostly concentrate on the nontrivial case, $s=+1$. Thus, our problem is a
two-point boundary value problem
\bea 
(ru_r)_r-\frac{n^2}r u+r(V+u^2)u&=&\beta ru,\label{2.1}\\
u(0)&=& u(R)=0,\label{2.2}
\eea
for which the parameter $\beta\in\bfR$ arises as an eigenvalue of the problem.

In order to approach the problem consisting of (\ref{2.1}) and (\ref{2.2}), we write down the action functional
\be \label{2.3}
I(u)=\frac12\int_0^R\left\{ ru^2_r+\frac{n^2}r u^2-rV(r) u^2-\frac r2u^4\right\}\,\dd r,
\ee
and the constraint functional
\be\label{2.4}
P(u)=\int|\psi|^2 r\,\dd r\dd\theta=2\pi\int_0^R ru^2\,\dd r,
\ee
which measures the beam power \cite{SK} of the vortex wave. Thus, to get a solution of (\ref{2.1})--(\ref{2.2}), it
suffices to prove the existence of a solution to the following constrained minimization problem
\be \label{min}
\min\left\{ I(u)\,|\,u\in {\cal A}, P(u)=P_0\right\},\quad P_0>0,
\ee 
where the admissible class $\cal A$ is defined by
\be 
{\cal A}=\left\{u(r)\mbox{ is absolutely continuous over }[0,R], \, u(0)=u(R)=0,\, E(u)<\infty\right\},
\ee
with
\be 
E(u)=\frac12\int_0^R\left\{ ru^2_r+\frac{1}r u^2+\frac r2u^4\right\}\,\dd r, 
\ee
being the `energy' functional, $P_0$ is a prescribed value for the beam power, and $\beta$ arises as the Lagrange 
multiplier. Note that the finite-energy condition $E(u)<\infty$ is imposed only to ensure that all the terms in
the indefinite action functional (\ref{2.3}) stay finite.

For convenience, for a function $f$ of the variable $r$, we interchangeably use $f_r$ and $f'$ to denote its derivative.
We will also need the following decomposition and notation
\be \label{V}
\left.\begin{array}{rll}V&=&V^+-V^-,\quad V^\pm=\max\{\pm V,0\}, \\
&&\\
V_0^\pm&=&\max\{V^\pm(r)\,|\, r\in[0,R]\},\\
&&\\
V_0&=&\max\{|V(r)|\,|\, r\in[0,R]\}=\max\{V^+_0,V^-_0\}.\end{array}\right\}
\ee

The main results of this section may be stated as follows.

\begin{theorem}\label{theorem1} For any nonzero integer $n$ and a given continuous potential function $V(r)$ over
the interval $[0,R]$ ($R>0$), consider 
the two-point boundary value problem (\ref{2.1})--(\ref{2.2}) governing an $n$-vortex
wave solution of the nonlinear Schr\"{o}dinger equation (\ref{1.1}), propagating along the $z$-axis with a propagation constant
$\beta$.

(i) The problem always has a solution pair $(u,\beta)$ with $u(r)>0$, $r\in(0,R)$, and $\beta\in\bfR$, so that the
associated beam power enjoys
the bound $P(u)<4\pi |n|$. In fact, such a solution may be obtained through solving the constrained minimization problem
(\ref{min}) assuming $P_0<4\pi|n|$, from which $\beta$ arises as a Lagrange multiplier.

(ii) Let $(u,\beta)$ be the solution pair of the problem (\ref{2.1})--(\ref{2.2}) obtained in part (i). Then $\beta$
has the lower bound
\be \label{Pii}
\beta\geq\frac{7P_0}{5\pi R^2}-V_0^- -\frac{12}{R^2}(1+n^2[2\ln2-1]).
\ee

(iii) Let $(u,\beta)$ be the solution pair of the problem (\ref{2.1})--(\ref{2.2}) obtained in part (i). Then $\beta<0$
if $|n|$ is sufficiently large so that
\be \label{Piii}
|n|>\left\{\frac{P_0^2}{4\pi^2}+\max\{r^2 V^+(r)\,|\,r\in[0,R]\}\right\}^{\frac12}.
\ee

(iv) The problem (\ref{2.1})--(\ref{2.2}) has no nontrivial small-beam-power solution satisfying $P(u)\leq\frac12$ if
the condition
\be 
n^2>r^2(V^+(r)-\beta),\quad r\in[0,R],
\ee
is fulfilled. So, roughly speaking,  the problem has no nontrivial small-power $P$ and small-propagation-constant (i.e., $|\beta|$
is sufficiently small) solution over
a small interval $[0,R]$.
\end{theorem}

We now establish these results.

(i) For any function $u$ satisfying $u(0)=0$, the Schwartz inequality implies that
\be \label{2.12}
u^2(r)=\int^r_0 2u(\rho)u_\rho(\rho)\,\dd\rho\leq 2\left(\int^r_0\rho u_\rho^2(\rho)\,\dd\rho\right)^{\frac12}
\left(\int_0^r\frac{u^2(\rho)}\rho\,\dd\rho\right)^{\frac12}.
\ee
Thus, multiplying (\ref{2.12}) by $r u^2$, integrating, and using $P(u)=P_0$, we have
\bea \label{2.13}
\int_0^R ru^4\,\dd r &\leq& \frac1\pi P_0\left(\int^R_0\rho u_\rho^2(\rho)\,\dd\rho\right)^{\frac12}
\left(\int_0^R\frac{u^2(\rho)}\rho\,\dd\rho\right)^{\frac12}\nn\\
&\leq&\vep \int^R_0\rho u_\rho^2(\rho)\,\dd\rho+\frac1\vep\left(\frac{P_0}{2\pi}\right)^2 \int_0^R\frac{u^2(\rho)}\rho\,\dd\rho.
\eea
Inserting (\ref{2.13}) into (\ref{2.3}), we obtain
\be \label{2.14}
I(u)\geq\frac12\left(1-\frac\vep2\right)\int_0^R r u_r^2\,\dd r+\frac12\left(n^2-\frac{P_0^2}{8\pi^2\vep}\right)\int_0^R
\frac{u^2}r\,\dd r-\frac1{4\pi} P_0 V_0.
\ee
In order to be able to find a suitable $\vep>0$ such that in (\ref{2.14}) we have
\be 
1-\frac\vep2>0,\quad n^2-\frac{P_0^2}{8\pi^2\vep}>0,
\ee
simultaneously, it suffices to assume that $P_0$ satisfies the condition
\be \label{2.16}
P_0<4\pi |n|.
\ee
In this situation, we can find two positive constants $C_1, C_2$, depending on $\vep, n, P_0$ but
independent of $u$, such that
\be\label{low}
I(u)\geq C_1 \int_0^R r u_r^2\,\dd r+C_2\int_0^R
\frac{u^2}r\,\dd r-\frac1{4\pi} P_0 V_0.
\ee

Assume  (\ref{2.16}) and let $\{u_m\}$ be a minimizing sequence of (\ref{min}). Then the coercive inequality (\ref{low}) 
gives us the bound
\be \label{bd}
\int_0^R r([u_m]_r)^2\,\dd r+\int_0^R \frac1r u_m^2\,\dd r\leq C,
\ee
where $C>0$ is a constant independent of $m$. 

 Since both functionals $I$ and $P$ are even, we have $I(u_m)\geq I(|u_m|)$ and
$P(u_m)=P(|u_m|)$, where we have also used the basic fact \cite{GT} that for any function $u$ its distributional
derivative must satisfy $||u|_r|\leq |u_r|$. In other words, we see that the sequence $\{u_m\}$ may be modified so that each $u_m$ is nonnegative,
$u_m\geq0$. Thus we may assume that the sequence $\{u_m\}$ consists of nonnegative-valued functions.
It is clear that these functions may be viewed as radially symmetric functions over the disk
$D_R=\{(x,y)\in\bfR^2\,|\, x^2+y^2\leq R^2\}$ which all vanish on $\pa D_R$.
Moreover, since (\ref{bd}) holds, we see immediately that $\{u_m\}$ is bounded under the 
radially symmetrically reduced norm $\|\,\|$ where
\be \label{H}
\|u\|^2=\int_0^R r u^2\,\dd r+\int_0^R r u_r^2\,\dd r,
\ee
for the standard Sobolev space $W^{1,2}_0(D_R)$ since 
\be \label{R22}
\int_0^R ru^2\,\dd r\leq R^2\int_0^R\frac1r u^2\,\dd r.
\ee
Hence we may assume without loss of generality that $\{u_m\}$ converges weakly to an element $u\in W^{1,2}_0(D_R)$
as $m\to\infty$. Applying the compact embedding $W^{1,2}(D_R)\to L^p(D_R)$ ($p\geq1$), 
we see that $u_m\to u$ strongly in $L^p(D_R)$ as
$m\to\infty$. Of course, $u$ is radially symmetric as well. Thus we may write it as $u=u(r)$ which satisfies $u(R)=0$.
Moreover, it is clear
that for any $\vep\in(0,R)$, $\{u_m\}$ is a bounded sequence in the space $W^{1,2}(\vep,R)$. Thus, using
the compact embedding $W^{1,2}(\vep,R)\to C[\vep,R]$, we see that $u_m\to u$ as $m\to\infty$ uniformly over
$[\vep, R]$. Besides, similar to (\ref{2.12}), we have for any pair $r_1,r_2\in (0,R)$, $r_1<r_2$, the inequality
\bea \label{19}
|u_m^2(r_2)-u_m^2(r_1)|&\leq &2\left(\int^{r_2}_{r_1}r ([u_m]_r)^2\,\dd r\right)^{\frac12}
\left(\int_{r_1}^{r_2}\frac{u_m^2}r\,\dd r\right)^{\frac12}\nn\\
&\leq&2C^{\frac12}
\left(\int_{r_1}^{r_2}\frac{u_m^2}r\,\dd r\right)^{\frac12},
\eea
where the constant $C>0$ is as given in (\ref{bd}). Letting $m\to\infty$ in (\ref{19}), we arrive at
\be \label{20}
|u^2(r_2)-u^2(r_1)|\leq 2C^{\frac12}\left(\int_{r_1}^{r_2}\frac{u^2}r\,\dd r\right)^{\frac12}.
\ee
However,
in view of (\ref{bd}) and Fatou's lemma, we have
\bea 
\int_0^R r u_r^2\,\dd r&\leq&\liminf_{m\to\infty}\int_0^R r([u_m]_r)^2\,\dd r,\\
\int_0^R \frac1r u^2\,\dd r&\leq&\liminf_{m\to\infty}\int_0^R \frac1ru_m^2\,\dd r.
\eea
In particular, in view of (\ref{bd}) again, we see that $\frac1r u^2\in L(0,R)$. Therefore the right-hand side of
(\ref{20}) tends to zero as $r_1,r_2\to 0$, which implies that the limit
\be 
\eta_0=\lim_{r\to0} u^2(r)
\ee
exists. Since $\frac1r u^2\in L(0,R)$, we must have $\eta_0=0$. Hence the boundary condition $u(0)=0$ is achieved.

Summarizing the above results, we conclude that the function $u$ obtained as the limit
of the minimizing sequence $\{u_m\}$ for the problem (\ref{min}) satisfies 
$u(0)=u(R)=0$, $u(r)\geq0$ for all $r\in[0,R]$, $E(u)<\infty$, and
\be 
I(u)\leq\liminf_{m\to\infty}I(u_m),\quad P(u)=\lim_{m\to\infty}P(u_m)=P_0.
\ee
Thus, $u$ is a solution to (\ref{min}). Consequently, there is some $\beta\in\bfR$ such that $(u,\beta)$
verify the boundary value problem (\ref{2.1})--(\ref{2.2}).

If there is a point $r_0\in (0,R)$ such that $u(r_0)=0$, then $u_r(r_0)=0$ since $r_0$ is a minimum point for
the function $u(r)$. Applying the uniqueness theorem for the initial value problem of ordinary differential equations,
we have $u(r)=0$ for all $r\in(0,R)$, which contradicts the fact $P(u)=P_0>0$. This proves $u(r)>0$ for all $r\in (0,R)$.

(ii) Let $(u,\beta)$ be a solution pair just obtained. We next study the quantity $\beta$ in (\ref{2.1}). 
As a preparation,
 we establish
\be \label{2.21}
\liminf_{r\to 0}\{ ru(r)|u_r(r)|\}=0.
\ee

Suppose otherwise that (\ref{2.21}) is not valid. Then there are some $\vep_0>0$ and
$r_0\in(0,R]$ such that
\be 
 ru(r)|u_r(r)|\geq \vep_0,\quad r\in(0,r_0),
\ee
which leads to
\be 
\infty=\int_0^{r_0}\frac{\vep_0}r\,\dd r\leq\int_0^{r_0} u|u_r|\,\dd r\leq \left(\int_0^{r_0}\frac1r u^2\,\dd r\right)^{\frac12}
\left(\int_0^{r_0}ru_r^2\,\dd r\right)^{\frac12},
\ee
which contradicts with $E(u)<\infty$. So (\ref{2.21}) is valid.

From (\ref{2.21}), we can find a sequence $\{r_m\}$ so that $r_m\to0$ as $m\to\infty$ and
\be \label{2.24}
\lim_{m\to\infty} \{r_m u(r_m) u_r(r_m)\}=0.
\ee

Multiplying (\ref{2.1}) by $u$, integrating over $[r_m,R]$, letting $m\to\infty$, and applying (\ref{2.24}), we obtain
\be \label{2.25}
\beta\int_0^R ru^2\,\dd r=\int_0^R(rVu^2+ru^4)\,\dd r-\int_0^R\left(\frac{n^2}r u^2+ru_r^2\right)\dd r.
\ee

Let $u_0$ be any absolutely continuous function satisfying $E(u_0)<\infty$, the boundary condition (\ref{2.2}), and
$P(u_0)=P_0$. Since $u$ solves (\ref{min}), we have $I(u)\leq I(u_0)$. As a consequence of this observation, we have
the bound
\be \label{2.26}
\int_0^R\left(\frac{n^2}r u^2+ru_r^2\right)\dd r\leq \int_0^R\left(rVu^2+\frac  r2 u^4\right)\dd r+2I(u_0).
\ee

Inserting (\ref{2.26}) into (\ref{2.25}), we obtain
\be \label{2.27}
\frac1{2\pi}\beta {P_0}\geq -2I(u_0)+\frac12\int_0^R ru^4\,\dd r.
\ee

To estimate the right-hand side of (\ref{2.27}), we set $R=2a$ for convenience and define
\be \label{2.28}
u_0(r)=\left\{\begin{array}{lrl}\frac bar,&& 0\leq r\leq a,\\ && \\\frac ba(2a-r),&&a\leq r\leq 2a.\end{array}\right.
\ee 
Therefore, after some calculation we have
\bea
P_0=2\pi\int_0^{2a} r u_0^2\,\dd r&=& \frac{4\pi}3a^2b^2,\label{2.28a}\\
\int_0^{2a} r (u_0'(r))^2\,\dd r&=&2b^2,\label{2.28b}\\
\int_0^{2a}\frac1r u_0^2\,\dd r&=&2b^2(2\ln2-1),\label{2.28c}\\
\int_0^{2a} r u^4_0\,\dd r&=&\frac25 a^2 b^4.\label{2.28d}
\eea

Using (\ref{2.28a})--(\ref{2.28d}) and (\ref{V}), we get
\bea 
I(u_0)&\leq&\frac12\int_0^R\left\{ r(u_0'(r))^2+\frac{n^2}r u_0^2-\frac r2u_0^4\right\}\,\dd r+\frac12\int_0^R r V^- u^2\,\dd r\nn\\
&\leq&b^2\left(1+n^2(2\ln2-1)+\frac13 V_0^- a^2-\frac1{10}a^2b^2\right).\label{2.33}
\eea
Besides, applying the Schwartz inequality, we have
\be \label{P2}
\left(\int_0^R ru^2\,\dd r\right)^2\leq\frac{R^2}2\int_0^R ru^4\,\dd r.
\ee
Thus, in view of (\ref{2.27}),  (\ref{2.33}),  and (\ref{P2}),  we arrive at
\be 
\frac1{2\pi}\beta P_0\geq2b^2\left(\frac{7}{45}a^2b^2-\left[\frac1{12}V_0^- R^2+1+n^2(2\ln2-1)\right]\right).
\ee
Inserting
$R=2a$ and $a^2b^2=3P_0/4\pi$,
we obtain the lower estimate for $\beta$:
\be \label{P3}
\beta\geq \frac{12}{R^2}\left(\frac{7}{60\pi}P_0-\left[\frac1{12}V_0^- R^2+1+n^2(2\ln2-1)\right]\right).
\ee

(iii) We next derive the sufficient condition stated to ensure $\beta<0$.

In fact, inserting (\ref{2.13}) into (\ref{2.25}), we have
\be \label{2.35}
\frac1{2\pi}\beta P_0\leq-(1-\vep)\int_0^R r u^2_r\,\dd r-\int_0^R\left(\frac{n^2}r-\frac{P_0^2}{4\pi^2\vep r}-rV\right) u^2\,\dd r.
\ee
For convenience, we may set $\vep=1$ in (\ref{2.35}). Thus, whenever the inequality
\be 
n^2> \frac1{4\pi^2}P_0^2 + r^2 V^+(r),\quad r\in[0,R],
\ee
is fulfilled, we can conclude with $\beta<0$ since $u(r)>0$ for $r\in(0,R)$.

(iv) We now consider nonexistence. For any admissible function $u$, we may view $u$ as a radially symmetric
function defined over $\bfR^2$ with support contained in the disk $D_R=\{(x,y)\in \bfR^2\,|\, x^2+y^2\leq R^2\}$. Hence, applying the classical Gagliardo--Nirenberg inequality over $\bfR^2$, we deduce
\be \label{GN}
\int_0^R ru^4\,\dd r\leq 4\pi\int_0^R ru^2\,\dd r\int_0^R ru_r^2\,\dd r,
\ee

From (\ref{2.25}) and inserting (\ref{GN}) with $P(u)=P_0$, we have
\bea 
0&=&\int_0^R(rVu^2+ru^4)\,\dd r-\int_0^R\left(\frac{n^2}r u^2+ru_r^2+\beta ru^2\right)\dd r\nn\\
&\leq&(2 P_0-1)\int_0^R ru^2_r\,\dd r-\int_0^R\left(\frac{n^2}{r^2}+\beta-V\right)r u^2\,\dd r.
\eea
Therefore, when the conditions
\be 
 P_0\leq \frac12,\quad \frac{n^2}{r^2}-V(r)+\beta>0,\quad r\in(0,R],
\ee
are imposed, we arrive at $u\equiv0$, as anticipated. Thus, in this situation, the problem 
consisting of (\ref{2.1})--(\ref{2.2}) has no nontrivial solution.

The proof of Theorem \ref{theorem1} is complete.

\section{Vortices as saddle points}
\setcounter{equation}{0}

In this section, we study the existence of optical vortices which are the solutions of the
boundary value problem (\ref{2.1})--(\ref{2.2}) as the saddle points of the action functional 
\be \label{2.3b}
I_\beta(u)=\frac12\int_0^R\left\{ ru^2_r+\frac{n^2}r u^2+(\beta-V(r))r u^2-\frac r2u^4\right\}\,\dd r,
\ee
with $|n|\geq1$.
We shall use a min-max theory approach. Suggested by the discussion of the previous section, we introduce the function
space $H$ which is the completion of the space $X=\{u\in C^1[0,R]\,|\, u(0)=u(R)=0\}$ (the set of differentiable 
functions over $[0,R]$ which vanish at the two endpoints of the interval) equipped with the inner product
\be 
\langle u,v\rangle_H=\int_0^R\left(r u_r v_r+\frac1r uv\right)\,\dd r,\quad u,v\in H.
\ee
As seen in the discussion of the previous section, 
as a Hilbert space, $H$ may be viewed as an embedded subspace of $W^{1,2}_0(D_R)$ consisting of radially
symmetric functions such that any element $u\in H$ enjoys the desired property $u(0)=0$.

In order to simplify the presentation of the study here, we assume that $\beta$ satisfies
\be \label{3.2}
\beta\geq V_0^+=\max\{V^+(r)\,|\, r\in[0,R]\}.
\ee

Recall that a $C^1$-functional $I:H\to \bfR$ is said to enjoy the Palais--Smale condition if for any sequence
$\{u_m\}$ satisfying the properties (i) $I(u_m)\to\alpha$ as $m\to\infty$, and (ii) $I'(u_m)\to 0$ as $m\to\infty$ as a sequence
in the dual space of $H$, one can extract a subsequence from $\{u_m\}$  which converges (strongly) in $H$.

As an initial step, we have 

\begin{lemma}\label{lemma3.1} The action functional $I_\beta$ given in (\ref{2.3b}) satisfies the Palais--Smale condition.
\end{lemma}
\begin{proof} 
It is straightforward to see that the functional (\ref{2.3b}) is $C^1$ over $H$.

Let $\{u_m\}$ be a sequence in $H$ satisfying the properties
\bea 
I_\beta(u_m)&=&\frac12\int_0^R\left\{ r([u_m]_r)^2+\frac{n^2}r u_m^2+\left(\beta-V(r)\right)r u_m^2-\frac r2 u_m^4\right\}\,\dd r\nn\\
&\to&\alpha,\quad m\to\infty,\label{3.3}\\
|I_\beta'(u_m)(v)|&=&\left|\int_0^R\left\{ r[u_m]_r v_r+\frac{n^2}r u_m v+(\beta-V(r))r u_m v-ru_m^3 v\right\}\,\dd r\right|\nn\\
&\leq&\vep_m\|v\|_H,\quad \vep_m\geq0,\quad v\in H,\label{3.4}
\eea
where $\vep_m\to0$ as $m\to\infty$. In (\ref{3.4}), we may take $v=u_m$ to get
\be \label{3.5}
\int_0^R ru_m^4\,\dd r\leq\int_0^R\left(r([u_m]_r)^2+\frac{n^2}r u_m^2+(\beta-V(r))ru_m^2\right)\,\dd r+\vep_m\|u_m\|_H.
\ee

On the other hand, in view of (\ref{3.3}), we may assume without loss of generality that $I_\beta(u_m)\leq \alpha+1$ for all $m
=1,2,\cdots$. Hence, applying (\ref{3.5}), the assumption (\ref{3.2}), and a simple interpolation inequality, we find
\bea 
2(\alpha+1)&\geq& \frac12\int_0^R\left(r([u_m]_r)^2+\frac{n^2}r u_m^2+(\beta-V(r)) ru_m^2\right)\,\dd r-\frac12\vep_m\|u_m\|_H\nn\\
&\geq&\frac14\|u_m\|^2_H-\frac14\vep_m^2,\quad m=1,2,\cdots.
\eea
In other words, $\{u_m\}$ is a bounded sequence in $H$. Without loss of generality, we may assume that $\{u_m\}$ weakly
converges to an element $u\in H$ as $m\to\infty$. It is clear that $u_m\to u$ as $m\to\infty$ strongly in any
$L^p(D_R)$ or $L^p((0,R), r\dd r)$ ($p\geq1$). Thus, letting $m\to\infty$ in (\ref{3.4}), we arrive at
\be \label{3.7}
\int_0^R\left\{ r u_r v_r+\frac{n^2}r u v+(\beta-V(r))r u v-r u^3 v\right\}\,\dd r=0,\quad\forall v\in H.
\ee
Let $v=u_m-u$ in (\ref{3.4}) and (\ref{3.7}) and insert the resulting (\ref{3.7}) into
the resulting (\ref{3.4}). We have
\bea 
&&\left|\int_0^R\left\{ r([u_m-u]_r)^2+\frac{n^2}r (u_m-u)^2+(\beta-V)r (u_m-u)^2-r(u_m^3-u^3)(u_m-u)\right\}\,\dd r\right|\nn\\
&&\leq \vep_m\|u_m-u\|_H.\label{3.8}
\eea
As a consequence of (\ref{3.8}) and $\beta\geq V_0^+$, we obtain
\be 
\|u_m-u\|^2_H\leq \vep_m\|u_m-u\|_H +\int_0^R r|(u_m^3-u^3)(u_m-u)|\,\dd r,\quad m=1,2,\cdots,
\ee
which immediately implies that $u_m\to u$ strongly in $H$ as $m\to\infty$, as desired.
\end{proof}

We next identify a mountain-pass structure through the following two lemmas.

\begin{lemma}\label{lemma3.2} There are constants $K>0$ and $C_0>0$ such that
\be \label{3.10}
\inf\{ I_\beta(u)\,|\, \|u\|_H^2=K\}\geq C_0.
\ee
\end{lemma}\begin{proof}
For any constant $K>0$, let $u\in H$ satisfy $\|u\|^2_H=K>0$. From (\ref{2.12}), we have
\be \label{3.11}
\int_0^R ru^4\,\dd r \leq 4\int_0^R r\,\dd r\int_0^R ru_r^2\,\dd r\int_0^R\frac{u^2}r\,\dd r\leq 2R^2 K^2.
\ee
Applying (\ref{3.11}) in (\ref{2.3b}), we find
\be \label{3.13}
I_\beta(u)\geq\frac12 \left(K- R^2 K^2\right)\equiv f(K).
\ee
However, the maximum of the function $f$  in (\ref{3.13}) is attained at $K_0=\frac1{2R^2}$ which gives us the value
$f(K_0)=\frac1{8R^2}$. So, in conclusion, we have the lower bound
\be 
I_\beta(u)\geq\frac1{8R^2},\quad \|u\|^2_H=\frac1{2R^2},
\ee
which establishes (\ref{3.10}).
\end{proof}

\begin{lemma}\label{lemma3.3} 
For any constant $K>0$, there is an element $v\in H$ satisfying $\|v\|_H^2>K$ and $I_\beta(v)<0$.
\end{lemma}
\begin{proof}
With $R=2a$, we will see that we
can use the function $u_0$ defined in (\ref{2.28}) as a test function. 

For this purpose, we first show that $u_0\in H$. To see this,
we need to prove that $u_0$ can be obtained in the limit from a sequence of functions in $X$ under the norm of $H$. In fact, for any $0<\vep<a$, we can define
\be \label{3.15}
u_\vep(r)=\left\{\begin{array}{rll}u_0(r),& r\in [0, a-\vep)\cup (a+\vep,2a],\\&\\ Q_\vep(r),&r\in [a-\vep,a+\vep],
\end{array}\right.
\ee
where $Q_\vep(r)$ is taken to be a quadratic function satisfying
\be 
Q_\vep(a\pm\vep)=u_0(a\pm\vep),\quad Q_{\vep}'(a\pm\vep)=u_0'(a\pm\vep).
\ee
Matching these conditions, we find
\be 
Q_\vep(r)=-\frac{b}{2a\vep}(r^2-2ar+[a-\vep]^2),\quad r\in[a-\vep,a+\vep],
\ee
which enjoys the bounds
\be \label{3.18}
\frac ba(a-\vep)\leq Q_\vep(r)\leq \frac ba\left(a-\frac\vep2\right),\quad | Q'_\vep(r)|\leq \frac ba,\quad r\in[a-\vep,
a+\vep].
\ee
Consequently, from (\ref{3.18}) we have
\be 
\lim_{\vep\to0}\int_{a-\vep}^{a+\vep}\left( r [Q'_\vep(r)]^2+\frac1r Q_\vep^2(r)\right)\,\dd r=0.
\ee
Therefore $\{u_\vep\}$ is a Cauchy sequence in $H$ as $\vep\to0$ whose limit is clearly $u_0$ in view of the definition
of $\{u_\vep\}$ given in (\ref{3.15}). This proves $u_0\in H$.

Using the results (\ref{2.28a})--(\ref{2.28d}), we have
\bea 
\|u_0\|^2_H &=& 4b^2\ln2,\label{3.20}\\
I_\beta(u_0)&\leq&\frac12\int_0^R\left\{ r([u_0]_r)^2+\frac{n^2}r u_0^2+(\beta+V^-) r u_0^2-\frac r2u_0^4\right\}\,\dd r\nn\\
&\leq&b^2\left(1+n^2(2\ln 2-1)+\frac13 a^2(\beta+V^-_0)-\frac1{10}a^2b^2\right).\label{3.21}
\eea
From (\ref{3.20}) and (\ref{3.21}), we see that for any $K>0$ we may choose $b>0$ sufficiently large
to get $I_\beta(u_0)<0$ and $\|u_0\|^2_H>K$.

Thus the lemma follows.
\end{proof} 

It is interesting to note that  (\ref{3.21}) implies that $I_\beta(u_0)\to-\infty$ as $b\to\infty$.
In other words, 
the action functional (\ref{2.3b}) is not bounded from below over $H$ which prevents a direct minimization approach
to the problem. Indeed, we are now prepared to obtain a nontrivial solution of the boundary value problem
(\ref{2.1})--(\ref{2.2}) as a saddle point of the functional (\ref{2.3b}) in the following theorem.

\begin{theorem} \label{theorem2} For any $\beta\geq V_0^+$, $|n|\geq1$, and $R>0$, the problem (\ref{2.1})--(\ref{2.2}) has a nontrivial solution
over the interval $[0,R]$.
Moreover, such a solution may be obtained from a min-max approach applied to the action functional (\ref{2.3b}).
\end{theorem}
\begin{proof}
Let $I_\beta$ be the action functional (\ref{2.3b}). Then Lemma \ref{lemma3.1} says that $I_\beta$ satisfies the Palais--Smale
condition. Let $K,C_0>0$ be the constant stated in Lemma \ref{lemma3.2}.
Using Lemma \ref{lemma3.3},  we can find some $u_0\in H$ such that $\|u_0\|_H^2>K$ and $I_\beta(u_0)<0$. Denote by $\Gamma$ the set of all 
continuous paths in $H$ that link
the zero element $0$ of $H$ to $u_0$:
\be 
\Gamma=\left\{g\in C([0,1];H)\,|\, g(0)=0, g(1)=u_0\right\}.
\ee
Therefore there is some point $t_g\in(0,1)$ such that $\|g(t_g)\|^2_H=K$. From the classical mountain-pass
theorem (cf. Evans \cite{Evans}), we know that
\be 
c=\inf_{g\in\Gamma}\max_{t\in[0,1]} I_\beta(g(t))\geq C_0,
\ee
is a critical value of $I_\beta$. In other words, there is an element $u\in H$ satisfying $I_\beta(u)=c$ which is a critical
point of $I_\beta$. Of course, $u$ is nontrivial. That is, $u$ cannot be the zero element of $H$.
\end{proof}

Recall that in Theorem \ref{theorem1}
the condition (\ref{Piii}) is obtained to ensure  the wave propagation constant $\beta$ to assume a negative
value. Although (\ref{Pii}) states a lower estimate for $\beta$, no condition has been obtained to ensure $\beta>0$
for the solution of the constrained minimization problem (\ref{min}). Theorem \ref{theorem2}, however, complements Theorem \ref{theorem1} in that it gives us a family of nontrivial solutions realizing arbitrarily
prescribed parameter $\beta$ in the entire interval $[V_0^+,\infty)$ for any vortex charge $|n|\geq1$ and $R>0$.
That is, our existence result indicates that $\beta$ may take any positive value above or at $V^+_0$.
\section{The defocusing case when $s=-1$}
\setcounter{equation}{0}

If  we have $s=-1$ in (\ref{1.1}) instead, then the action functional (\ref{2.3}) is replaced by
\be \label{3.27}
I(u)=\frac12\int_0^R\left\{ ru^2_r+\frac{n^2}r u^2-rV(r) u^2+\frac r2u^4\right\}\,\dd r,
\ee
the difficulty with the quartic term, which was negative before, disappears, and the constrained minimization
problem (\ref{min}) is easily solved, which gives us a solution to the associated equation
\be 
(ru_r)_r-\frac{n^2}r u+r(V-u^2)u=\beta ru,\label{3.38}
\ee 
for some $\beta\in\bfR$. As before, this equation leads us to the relation
\be \label{3.29}
\int_0^R (\beta-V^+)ru^2\,\dd r=-\int_0^R\left(ru_r^2+\frac{n^2}r u^2+rV^-u^2+ru^4\right)\dd r.
\ee
Thus, if $\beta$ satisfies $\beta\geq V^+_0$, then $u\equiv0$. Therefore, regardless of the value of $R$, the problem
prevents the existence of a nontrivial solution for sufficiently large propagation constant $\beta$. This conclusion
is in sharp contrast to that in the case when $s=+1$ stated in Theorem \ref{theorem2}.

In general, the simple relation (\ref{3.29}) clearly indicates that it is natural for $\beta$ to take negative rather than
positive values. For example, using the Poincar\'{e} inequality over $D_R$,
\be 
\int_0^R r u^2\,\dd r\leq \left(\frac{R}{r_0}\right)^2\int_0^R ru^2_r\,\dd r,
\ee
where $r_0$ ($\approx 2.404825$) is the first positive zero of the Bessel function $J_0$, and (\ref{R22}),
we obtain from (\ref{3.29}) 
 the result
\bea \label{3.31}
\beta\int_0^R ru^2\,\dd r&\leq& V_0^+\int_0^R ru^2\,\dd r-\int_0^R\left(ru_r^2+\frac{n^2}r u^2+rV^-u^2+ru^4\right)\dd r\nn\\
&\leq&-\left(\left[\frac{r_0}R\right]^2+\left[\frac nR\right]^2 -V_0^+\right)\int_0^R ru^2\,\dd r-\int_0^Rru^4\,\dd r.
\eea
Consequently, we obtain
\be 
\beta<-\left(\frac{r_0^2+n^2}{R^2}-V_0^+\right).
\ee
In particular, we have $\beta\to-\infty$ as $|n|\to\infty$. 

An example of the defocusing case $s=-1$ is the study carried out in \cite{K} where $V$ takes the form $V=p J^2_1(br)$
given in terms of the Bessel function $J_1$ and  positive parameters $p,b$. Thus $V^-\equiv 0$. Another example of
the case $s=-1$ is in the lines of the studies \cite{OK,SMF}. There, although $V$ is not radially symmetric, it is non-positive valued,
$V^+\equiv0$. Here, assuming $V$ is radial as well as non-positive, then (\ref{3.29}) indicates that $\beta<0$ is the only possibility.
\medskip 

Note that the ($z$-direction)
angular momentum of the obtained stationary vortex wave in view of (\ref{1.2}) has the simple but
elegant expression \cite{SK}
\be 
L_z=\mbox{Im}\int (\psi^* \pa_\theta\psi) \,r\dd r\dd\theta=2\pi n\int_0^R r u^2\,\dd r=n P.
\ee

Write $\phi(x,y)=u(r)\e^{\ii n\theta}$, where $u$ solves (\ref{1.3}) and satisfies $u(0)=0$, and $\Delta=\nabla^2_\perp$. Then $\phi$ satisfies
\be 
\Delta\phi+(V+s|\phi|^2)\phi=\beta\phi,
\ee
away from the origin of $\bfR^2$. The condition $u(0)=0$ ensures that the origin is a removable singularity
\cite{Ybook} such that when $V$ is an analytic function of $(x,y)$, so is $\phi$. Consequently, in this
situation $u$ vanishes
at $r=0$ like $r^n$ for an $n$-vortex solution as in the classical Ginzburg--Landau equation case \cite{de,JT,Plohr}.

For the focusing case $s=+1$ with a non-positive potential (cf. \cite{YM}), we have $V^+\equiv0$ and  
the statements of our results simplify considerably. For example, for the solution pair $(u,\beta)$ obtained in
Theorem \ref{theorem1} to have the property $\beta<0$, it suffices that the vortex number $n$ satisfies 
the condition
\be 
|n|>\frac{P_0}{2\pi}.
\ee
Moreover, applying Theorem \ref{theorem1} (iv), we see that there is no nontrivial solution satisfying
\be 
P(u)\leq\frac12\quad\mbox{and}\quad n^2>-r^2\beta,\quad r\in[0,R].
\ee
In particular, we conclude that there is no nontrivial solution with $P(u)\leq\frac12$ and $\beta\geq0$.
Besides, in this case Theorem \ref{theorem2} becomes an existence theory for any prescribed propagation
constant $\beta\geq0$.

\small{

}


\begin{thebibliography}{99}

\bibitem{A}
S. K. Adhikari,
Localization of a Bose--Einstein condensate vortex in a bichromatic optical lattice,
{\em Phys. Rev.} A {\bf81} (2010) 043636.
 
\bibitem{AB}
L. Allen, M. W. Beijersbergen, R. J. C. Spreeuw, and J. P. Woerdman, 
Orbital angular momentum of light and the transformation of Laguerre--Gaussian laser modes,
 {\em Phys. Rev.} A {\bf45} (1992) 8185--8189.


\bibitem{BKK}
M. L. M. Balistreri, J. P. Korterik, L. Kuipers, and N. F. van Hulst,
Local observations of phase singularities in optical fields in waveguide structures,
{\em Phys. Rev. Lett.} {\bf85} (2000) 294--297.

\bibitem{BSV}
A. Bekshaev, M. Soskin, and M. Vasnetsov,
Paraxial light beams with angular momentum,
{\em Ukrainian J. Phys.} {\bf 2} (2005) 73--113.

\bibitem{Chiao}
R. Y. Chiao, E. Garmire, and C. H. Townes,
Self-trapping of optical beams, {\em Phys. Rev. Lett.} {\bf13} (1964)
 479--482.

\bibitem{CG}
J. E. Curtis and D. G. Grier,
Structure of optical vortices,
 {\em Phys. Rev. Lett.} {\bf90} (2003) 133901.

\bibitem{DY}
T. A. Davydova and A. I. Yakimenko,
Stable multi-charged localized optical vortices in cubic–quintic nonlinear media,
 {\em J. Optics} A {\bf97} (2004) S197--S201.

\bibitem{de}
H. J. de Vega and F. A. Schaposnik,
Classical vortex solution of the Abelian Higgs model, {\em Phys. Rev.} D
{\bf14} (1976) 1100--1106.

\bibitem{DK}
M. R. Dennis, R. P. King, B. Jack, K. O'Holleran, and M. J. Padgett,
Isolated optical vortex knots,
{\em Nature Phys.} {\bf6} (2010) 118--121.

\bibitem{DKT}
A. S. Desyatnikov, Y. S. Kivshar, and L. Torner, 
Optical vortices and vortex solitons,  {\em Progress in Optics} {\bf47} (2005) 291--391.

\bibitem{Du}
Z. Dutton and J. Ruostekoski,
Transfer and storage of vortex states in light and matter waves,
{\em Phys. Rev. Lett.} {\bf93} (2004) 193602.
 
\bibitem{Evans}
L. C. Evans, {\em Partial Differential Equations}, Amer. Math. Soc., Providence, 2002.

\bibitem{GT}
D. Gilbarg and N. Trudinger, {\em Elliptic Partial Differential Equations of
Second Order}, Springer, Berlin and New York, 1977.

\bibitem{JT}
A. Jaffe and C. H. Taubes, {\em Vortices and Monopoles}, Birkh\"{a}user, Boston, 1980.

\bibitem{KK}
A. M. Kamchatnov and
    S. V. Korneev,
Dynamics of ring dark solitons in Bose-–Einstein condensates and nonlinear optics,
{\em Phys. Lett.} A {\bf374} (2010) 4625--4628.

\bibitem{KMT}
Y. V. Kartashov,
B. A. Malomed,
and L. Torner, Solitons in nonlinear lattices,
 {\em Rev. Mod. Phys.} {\bf83} (2011) 247--305.


\bibitem{KVT}
Y. V. Kartashov, V. A. Vysloukh, and Lluis Torner,
Rotary solitons in Bessel optical lattices,
{\em Phys. Rev. Lett.} {\bf93} (2004)
093904. 

\bibitem{K}
Y. V. Kartashov, V. A. Vysloukh, and L. Torner,
Stable ring vortex solitons in Bessel optical lattices,
{\em Phys. Rev. Lett.} {\bf94} (2005) 043902. 

\bibitem{LD}
J. Leach, M. R. Dennis, J. Courtial, and M. J. Padgett,
Laser beams:  knotted threads of darkness,
{\em Nature} {\bf432} (2004) 165.

\bibitem{MSZ}
A. V. Mamaev,  M. Saffman, and
A. A. Zozulya,
Propagation of dark stripe beams in nonlinear media: snake instability and
creation of optical vortices, {\em Phys. Rev. Lett.} {\bf76} (1996) 2262--2265.

\bibitem{NAO}
D. Neshev, T. J. Alexander, E. A. Ostrovskaya, Y. S. Kivshar, H. Martin, I. Makasyuk, and Z. Chen,
Observation of discrete vortex solitons in optically-induced photonic lattices,
{\em Phys. Rev. Lett.} {\bf92} (2004) 123903.

\bibitem{NNK}
D. Neshev, A. Nepomnyashchy, and Yu. S. Kivshar,
 Nonlinear Aharonov--Bohm scattering by optical vortices,
{\em Phys. Rev. Lett.} {\bf87} (2001) 043901.

\bibitem{NB}
 J. F. Nye and M. V. Berry,  Dislocations in wave trains, {\em Proc. Roy. Soc.}  A {\bf336} (1974) 165--190.

\bibitem{OK}
E. A. Ostrovskaya and Y. S. Kivshar, 
Matter-wave gap vortices in optical lattices,
 {\em Phys. Rev. Lett.} {\bf93} (2004) 160405.


\bibitem{Plohr}
B. J. Plohr, The existence, regularity, and behavior of isotropic solutions
of classical gauge field theories, Thesis, Princeton University,
1980.

\bibitem{RLS}
D. Rozas, C. T. Law, and G. A. Swartzlander, Jr., Propagation dynamics of optical vortices, {\em J. Optical Soc. Amer.}
B {\bf14} (1997) 3054--3065.

\bibitem{RSS}
D. Rozas, Z. S. Sacks, and G. A. Swartzlander, Jr.,
Experimental observation of fluid-like motion of optical vortices, {\em Phys. Rev. Lett.}
{\bf79} (1997) 3399--3402.

\bibitem{SK}
J. R. Salgueiro and Y. S. Kivshar, Switching with vortex beams in nonlinear concentric couplers,
{\em Opt. Exp.} {\bf20} (2007) 12916--12921.

\bibitem{SO}
J. Scheuer and M. Orenstein,  Optical vortices crystals: spontaneous generation in nonlinear semiconductor
microcavities, {\em Science} {\bf285} (1999) 230--233.

\bibitem{SMF}
R. G. Scott, A. M. Martin, T. M. Fromhold, S. Bujkiewicz, F.W. Sheard, and M. Leadbeater,
Creation of solitons and vortices by Bragg reflection of Bose--Einstein condensates in an optical lattice,
{\em Phys. Rev. Lett.} {\bf90} (2003) 110404. 

\bibitem{SGV}
M. S. Soskin, V. N. Gorshkov, and M. V. Vasnetsov,
Topological charge and angular momentum of light beams carrying optical vortices,
{\em Phys. Rev.} A {\bf56} (1997) 4064--4075.

\bibitem{SL}
G. A. Swartzlander, Jr. and C. T. Law,
Optical vortex solitons observed in Kerr nonlinear media,
{\em Phys. Rev. Lett.} {\bf69} (1992)
2503--2506. 

\bibitem{VB}
A. Vinçotte and L. Berge,
Femtosecond optical vortices in air,
 {\em Phys. Rev. Lett.} {\bf95} (2005) 193901.

\bibitem{YM}
J. Yang and Z. H. Musslimani,
Fundamental and vortex solitons in a two-dimensional optical lattice,
{\em Optics Lett.} {\bf28} (2003) 2094--2096.

\bibitem{Ybook}
Y. Yang, {\em Solitons in Field Theory and Nonlinear Analysis}, Springer-Verlag, New York, 2001.

\end{thebibliography}
\end{document}